\documentclass{LMCS}

\def\Prob{\mathop{\it Prob}\nolimits}

\usepackage{latexsym} 
\usepackage{amssymb,amsmath}
\usepackage{epsfig}
\usepackage{enumerate,hyperref}

\def\nat{{\mathbb N}}
\def\real{{\mathbb R}}

\newcommand{\support}{\mbox{\tt support}}

\theoremstyle{plain}
\newtheorem{theorem}[thm]{Theorem}

\newtheorem{corollary}[thm]{Corollary}
 
\newtheorem{lemma}[thm]{Lemma}
\newtheorem{proposition}[thm]{Proposition} 
\newcommand{\rfotwo}%                    % FO with two variables and < only
{\mbox{\rm FO}^2[<]}
\newcommand{\tl}{\mathrm{TL}}            % temporal logic

\newcommand{\utl}%                       % unary TL
{\mbox{{\rm unary-}}\tl}

\title{Multi-Objective Model Checking of \\Markov Decision
Processes\rsuper *}

\titlecomment{{\lsuper *}A preliminary version of this paper 
appeared in the \emph{Proceedings of the 13th International 
Conference on Tools and Algorithms for the Construction and Analysis
of Systems ({TACAS}'07).}}

\author[K. Etessami]{Kousha Etessami\rsuper a}
\address{{\lsuper a}LFCS, School of Informatics, University of Edinburgh, UK}
\email{kousha@inf.ed.ac.uk}

\author[M. Kwiatkowska]{Marta Kwiatkowska\rsuper b}
\address{{\lsuper b}Computing Laboratory, Oxford University, UK}
\email{Marta.Kwiatkowska@comlab.ox.ac.uk}

\author[M. Y. Vardi]{Moshe Y. Vardi\rsuper c}
\address{{\lsuper c}Department of Computer Science, Rice University, USA}
\email{vardi@cs.rice.edu}

\author[M. Yannakakis]{Mihalis Yannakakis\rsuper d}
\address{{\lsuper d}Department of Computer Science, Columbia University, USA}
\email{mihalis@cs.columbia.edu}


\def\doi{4 (4:8) 2008}
\lmcsheading%
{\doi}
{1--21}
{}
{}
{Oct.~\phantom{0}4, 2007}
{Nov.~12, 2008}
{}   

\begin{document}

\maketitle

\begin{abstract}
We study and provide efficient algorithms for 
multi-objective model checking problems for 
Markov Decision Processes (MDPs).  
Given an MDP, $M$, and
given multiple linear-time ($\omega$-regular or LTL) properties 
$\varphi_i$, 
and probabilities $r_i \in [0,1]$,
$i=1,\ldots,k$,
we ask whether there exists a strategy $\sigma$ for the controller
such that, for all $i$, the probability that
a trajectory of $M$ controlled by $\sigma$
satisfies $\varphi_i$ is at least $r_i$.
We provide an  algorithm that decides whether there exists 
such a strategy and if so produces it, and which
runs in time polynomial in the size of the MDP. 
Such a strategy may require the use of both randomization and memory.
We also consider more general multi-objective $\omega$-regular queries,
which we motivate with 
an application to assume-guarantee compositional reasoning 
for probabilistic systems.

Note that there can be trade-offs between different properties:
satisfying property $\varphi_1$ with high probability
may necessitate satisfying $\varphi_2$ with low probability. 
Viewing this as a multi-objective optimization problem,
we want information about 
the ``trade-off curve'' or {\em Pareto curve} for maximizing the 
probabilities of different properties. 
We show that one can compute an
approximate
Pareto curve with respect to a set of
$\omega$-regular properties in time polynomial 
in the size of the MDP.

Our quantitative upper bounds use LP methods.
We also study qualitative multi-objective model checking problems,
and we show that these can be analysed by purely graph-theoretic
methods, even though the strategies may still require 
both randomization and memory. 
\end{abstract}

\vfill\eject
\section{Introduction}

Markov Decision Processes (MDPs) are standard models for stochastic
optimization and for 
modelling systems with probabilistic and nondeterministic or controlled behavior 
(see \cite{Puterman94,Var85,CY95,CY98}).  
In an MDP, at each
state, the controller can choose from among a number of
actions, or choose a probability distribution over actions.  Each action
at a state
determines a probability distribution on the next state.
Fixing an initial state and fixing the 
controller's strategy determines 
a probability space of infinite runs (trajectories) of the MDP.  
For MDPs with a single objective, 
the controller's goal is to
optimize the value of an objective function, or payoff, which 
is a function of the entire trajectory.  
Many different objectives have been studied for MDPs,
with a wide variety of applications.  In particular, in
verification research  linear-time model checking of
MDPs has been studied, where the objective is to maximize the
probability that the trajectory satisfies a given 
$\omega$-regular or LTL property (\cite{CY98,CY95,Var85}).

In many settings we may not just care about a single
property.   Rather, we may have a number of different properties
and we may want to know whether we can simultaneously satisfy
all of them with given probabilities.
For example, in a system with a server and
two clients, we may want to maximize the
probability for both clients $1$ and $2$ of the temporal property:
``every request issued by client $i$ eventually
receives a response from the server'', $i=1,2$.  
Clearly, there may be a trade-off.  To increase this
probability for client $1$ we may have to decrease it for client $2$,
and vice versa.  We thus want to know what are the
simultaneously {\em achievable} pairs $(p_1,p_2)$ 
of probabilities for the two properties.
More specifically, we will be interested in the ``trade-off curve''
or {\em Pareto curve}.  The Pareto curve is the set of all
achievable vectors $p= (p_1,p_2) \in [0,1]^2$ 
such that there does not 
exist another achievable vector $p'$ that  {\em dominates} $p$, meaning
that $p \leq p'$ (coordinate-wise inequality) and $p \neq p'$.

\begin{figure*}[t]
\begin{center}
\scalebox{0.50}{\input{example_2mdp.pstex_t}}
\caption{An MDP with two objectives, $\Diamond P_1$ and $\Diamond P_2$, and the associated Pareto curve.}
\label{fig:example_2mdp}
\end{center}
\end{figure*}

Concretely, consider the very simple MDP depicted in Figure \ref{fig:example_2mdp}.
Starting at state $s$, we can take one of three possible
actions $\{a_1, a_2,a_3\}$.
Suppose we are interested in LTL
properties $\Diamond P_1$ and 
$\Diamond P_2$.
Thus we want to maximize 
the probability of reaching the two distinct vertices labeled by
$P_1$ and $P_2$, respectively.  
To maximize the
probability of $\Diamond P_1$ we should take action $a_1$,
thus reaching $P_1$ with probability $0.6$ and $P_2$ with probability $0$.
To maximize the probability of $\Diamond P_2$ we should take $a_2$,
reaching $P_2$ with probability $0.8$ and $P_1$ with probability $0$.
To maximize the {\em sum} total probability of reaching $P_1$ or 
$P_2$, we should take $a_3$, reaching both with probability $0.5$.
Now observe that we can also ``mix'' these pure 
strategies using randomization
to obtain any convex combination of these three value vectors.
In the graph on the right in Figure \ref{fig:example_2mdp}, the dotted line 
plots the Pareto curve for these two properties.

The Pareto curve ${\mathcal P}$ in general contains infinitely many points,
and it can be too costly to compute an exact representation for it
(see Section \ref{sec:basics}).
Instead of computing it outright we can try to {\em approximate} it 
(\cite{PapaYan00}).  An {\em $\epsilon$-approximate Pareto curve} is a set of 
achievable
vectors ${\mathcal P}(\epsilon)$ such that for every achievable vector $r$ 
there is some vector $t \in {\mathcal P}(\epsilon)$ which ``almost'' dominates it,
meaning  $r \leq (1+ \epsilon) t$.

In general, given a labeled MDP $M$, 
$k$ distinct $\omega$-regular properties, 
$\Phi = \langle \varphi_i \mid i=1,\ldots,k\rangle$,
a start state $u$, and a strategy $\sigma$, 
let $\Pr^\sigma_u(\varphi_i)$
denote the probability that starting at $u$, using strategy $\sigma$,
the trajectory satisfies $\varphi_i$.
For a strategy $\sigma$, define the 
vector $t^\sigma = ( t^\sigma_1,\ldots, t^\sigma_k)$,
where $t^\sigma_i = \Pr^\sigma_u(\varphi_i)$, for $i=1,\ldots,k$.
We say a value vector $r \in [0,1]^k$
is {\em achievable} for $\Phi$,
if there exists a strategy $\sigma$ such that $t^\sigma \geq r$.

We provide an algorithm that given MDP $M$, start state $u$, 
properties $\Phi$, 
and rational value vector $r \in [0,1]^k$,
decides whether $r$ is achievable, 
and if so produces a strategy $\sigma$
such that $t^{\sigma} \geq r$.
The algorithm runs in time polynomial in the size of the
MDP.  The strategies may require both randomization and memory.
Our algorithm works by first reducing the achievability problem
for multiple 
$\omega$-regular properties to  
one with  multiple reachability objectives,
and then reducing the multi-objective reachability problem 
to a multi-objective linear programming problem.
We also show that one can compute an 
$\epsilon$-approximate Pareto
curve for $\Phi$ in time polynomial
in the size of the MDP and in $1/\epsilon$.
To do this, we use our linear programming characterization
for achievability, and 
use results from \cite{PapaYan00} on approximating 
the Pareto curve for multi-objective linear programming problems.

We also consider more general {\em multi-objective queries}.
Given a boolean combination $B$ of quantitative predicates of the
form $\Pr^\sigma_u(\varphi_i) \Delta p$, where $\Delta \in \{ \leq
,\geq , < , > , = , \neq \}$, and $p \in [0,1]$,  a {\em multi-objective
query} asks whether there exists a strategy $\sigma$ satisfying 
$B$ (or whether {\em all} strategies $\sigma$
satisfy $B$).  
It turns out that 
such queries 
are not really much more
expressive than checking achievability.  Namely, 
checking a fixed query $B$
can be reduced to checking a fixed number of {\em extended achievability}
queries, where for some of the coordinates $t^\sigma_i$ we
can ask for a strict inequality, i.e., that $t^\sigma_i > r_i$.
(In general, however, the number and size of the
extended achievability queries needed 
may be exponential  in the size of $B$.)
A motivation for allowing
general multi-objective queries is to
enable {\em assume-guarantee compositional reasoning} for probabilistic 
systems,
as explained in Section \ref{sec:basics}.

Whereas our algorithms for quantitative problems use LP methods,
we also consider qualitative  multi-objective queries.
These are queries given by boolean combinations of 
predicates of the form $\Pr^\sigma_u(\varphi_i) \Delta b$,
where $b \in \{0,1\}$.   
We give an algorithm  using purely graph-theoretic techniques 
that decides whether there is a strategy
that satisfies a qualitative multi-objective query, 
and if so produces such a strategy.
The algorithm runs in
time polynomial in the size of the MDP.
Even for satisfying qualitative queries the strategy may need to  
use both randomization and memory.

In typical applications, the MDP is far larger than the 
size of the query.
Also, $\omega$-regular properties can be presented in many ways,
and it was already shown in \cite{CY95} that 
the query complexity of model
checking MDPs against even a single LTL property is 2EXPTIME-complete.
We remark here that, if properties are expressed via LTL
formulas, then our algorithms run in polynomial time 
in the size of the MDP and in 2EXPTIME in the size of the query,
for deciding arbitrary multi-objective
queries, where both the MDP and the query are part of the input.
So, the worst-case upper bound is the same as with a single LTL objective.
However, to keep  our complexity analysis simple, we
focus in this paper on the model complexity of our algorithms,
rather than their query complexity or combined complexity.

\noindent {\bf Related work.}  
Model checking of MDPs with a single $\omega$-regular objective
has been studied in detail (see \cite{CY98,CY95,Var85}).
In \cite{CY98}, Courcoubetis and Yannakakis also considered
MDPs with a single objective given by a positive weighted sum of
the probabilities of multiple $\omega$-regular properties,
and they showed how to efficiently optimize such objectives for MDPs.
They did not consider tradeoffs between multiple 
$\omega$-regular objectives.  We employ and build on techniques
developed in \cite{CY98}.

Multi-objective optimization is
a subject of intensive study
in Operations Research and related fields 
(see, e.g., \cite{Ehrgott05,Clima97}).  
Approximating the Pareto curve for general multi-objective optimization
problems was considered by Papadimitriou and Yannakakis in  
\cite{PapaYan00}.
Among other results, 
\cite{PapaYan00} showed that for multi-objective linear programming 
(i.e., linear constraints
and multiple linear objectives), one can 
compute a (polynomial sized) $\epsilon$-approximate Pareto curve 
in time polynomial in the size of the LP and in $1/\epsilon$.

Our work is related to recent work by Chatterjee, Majumdar, and Henzinger
(\cite{CMH06}), who 
considered MDPs with multiple discounted reward objectives.
They showed that randomized but memoryless strategies suffice for 
obtaining any achievable value vector for these objectives, and 
they reduced the multi-objective optimization and achievability 
(what they call {\em Pareto realizability}) problems 
for MDPs with discounted rewards to multi-objective
linear programming.  They were thus able to apply the results
of \cite{PapaYan00} in order to approximate the Pareto curve
for this problem.
We work in an undiscounted setting, where objectives can 
be arbitrary $\omega$-regular properties.
In our setting, strategies may require both randomization and memory 
in order to achieve a given value vector.
As described earlier, our algorithms first reduce multi-objective
$\omega$-regular problems to multi-objective reachability problems,
and we then solve multi-objective reachability problems 
by reducing them to multi-objective LP.
For multi-objective reachabilility,
we show randomized memoryless strategies do suffice.
Our LP methods for multi-objective reachability
are closely related to the LP methods used in \cite{CMH06} 
(and see also, e.g., \cite{Puterman94}, Theorem 6.9.1.,  where 
 a related result about discounted MDPs is established).
However, in order to establish the results in our undiscounted setting,
even
for reachability we have to overcome some new obstacles
that do not arise in the discounted case.  In particular,
whereas the ``discounted frequencies'' used in \cite{CMH06} 
are always well-defined finite
values under all strategies, the analogous undiscounted 
frequencies or
``expected number of visits'' can in general be infinite
for an arbitrary strategy.
This forces us to preprocess the MDPs in such a way that ensures 
that a certain family of
undiscounted stochastic flow equations has a 
finite solution which corresponds to the ``expected number of visits''
at each state-action pair under a given (memoryless) strategy. 
It also forces us to give
a quite different proof that memoryless strategies suffice
to achieve any achievable vector for multi-objective reachability, based 
on the convexity of the memorylessly achievable set. 

Multi-objective MDPs have also been studied
extensively in the OR and stochastic control 
literature (see e.g. \cite{Furukawa80,White82,Henig83,Ghosh90,WakTog98}).
Much of this work is typically concerned with  
discounted reward
or long-run average reward models, 
and does not focus on 
the complexity of algorithms.  
None of this work seems to directly imply even our 
result that for multiple reachability objectives checking achievability 
of a value vector can 
be decided in polynomial time, 
not to mention
the more general results for multi-objective model checking.

\vspace*{-0.1in}

\section{Basics and background}
\label{sec:basics}

A finite-state MDP $M = (V,\Gamma,\delta)$ consists of a finite
set $V$ of states, an action alphabet $\Gamma$,
and a transition relation $\delta$.
Associated with each state $v$ is a set of 
enabled actions $\Gamma_v \subseteq \Gamma$.
The transition relation is given by
$\delta \subseteq V \times \Gamma \times [0,1] \times V$.
For each state $v \in V$, each enabled action $\gamma \in \Gamma_v$,
and every state $v' \in V$,
we have at most one transition $(v,\gamma,p_{(v,\gamma,v')},v') \in \delta$,
for some probability $p_{(v,\gamma,v')} \in (0,1]$,  
such that $\sum_{v' \in V} p_{(v,\gamma,v')} = 1$. 
When there is no transition $(v,\gamma,p_{(v,\gamma,v')},v')$,
we may, only for notational convenience, sometimes assume that
there is a probability $0$ transition, i.e., that $p_{(v,\gamma,v')} = 0$.
(But such redundant probability $0$ transitions 
 are not part of the actual input.)
Thus, at each state, each enabled action determines a probability distribution
on the next state.
There are no other transitions, so no transitions on disabled actions.  
We assume every state $v$ has some enabled action,
i.e., $\Gamma_v \neq \emptyset$, so there are no dead ends.
For our complexity analysis, we assume of course that 
all probabilities $p_{(v,\gamma,v')}$
are rational.  
There are other ways to present MDPs, e.g., by
separating controlled
and probabilistic nodes into distinct states. 
The different presentations are equivalent and efficiently translatable 
to each other.

A labeled MDP $M= (V,\Gamma,\delta,l)$ has, in addition,  
a
set of propositional predicates $Q = \{ Q_1, \ldots, Q_r\}$ 
which label the states.
We view this as being given by a labelling function $l: V \mapsto \Sigma$, where $\Sigma = 2^Q$.
We define the encoding size of a (labeled) MDP $M$, denoted by $|M|$,
to be the total size required to encode 
all transitions and their
rational probabilities, where rational values are encoded
with numerator and denominator given in binary, as well as all
state labels.

For a labeled MDP $M = (V,\Gamma,\delta,l)$ with a given initial 
state $u \in V$,  which we denote by $M_u$, 
runs of $M_u$ are infinite sequences of states 
$\pi = \pi_0 \pi_1 \ldots \in V^\omega$, where $\pi_0 = u$ and
for all $i \geq 0$, $\pi_i \in V$ 
and there is a transition
$(\pi_i, \gamma, p, \pi_{i+1}) \in \delta$,
for some $\gamma \in \Gamma_{\pi_i}$ and some probability $p > 0$.  
Each run induces an $\omega$-word
over $\Sigma$, namely $l(\pi) 
\doteq l(\pi_0) l(\pi_1) \ldots \in \Sigma^\omega$.

A {\em strategy} is a function $\sigma: (V \Gamma)^* V \mapsto {\mathcal D}(\Gamma)$, 
which maps a finite history of play 
to a probability distribution on the  
next action.
Here ${\mathcal D}(\Gamma)$ denotes the set of probability distributions on
the set $\Gamma$.
Moreover, it must be the case that for all 
histories $wu$, 
$\sigma(wu) \in {\mathcal D}(\Gamma_u)$, i.e., the probability distribution
has support only over the actions available at state $u$.
A strategy is {\em pure} if $\sigma(wu)$ has support on exactly
one action, i.e., with probability 1 
a single action is played at every history.  
A strategy is {\em memoryless} (stationary) if the
strategy depends only on the last state, i.e., 
if $\sigma(wu) = \sigma(w'u)$ for all
$w, w' \in (V\Gamma)^*$.   If $\sigma$ is memoryless, we can simply define
it as a function $\sigma: V \mapsto {\mathcal D}(\Gamma)$.
An MDP $M$ with initial state $u$, together with a strategy $\sigma$,
naturally induces a Markov chain $M^\sigma_u$, whose states
are the histories of play in $M_u$, and
such that from state
$s = wv$  if $\gamma \in \Gamma_v$, there is a transition
to state $s' = wv\gamma v'$ with probability 
$\sigma(wv)(\gamma) \cdot p_{(v,\gamma,v')}$.
A run $\theta$ in $M^\sigma_u$
is thus given by a sequence $\theta = \theta_0 \theta_1 \ldots$, where 
$\theta_0 = u$ and each
$\theta_i \in (V\Gamma)^*V$, for all $i \geq 0$.
We associate to each history $\theta_i = w v$ the 
label of its last state $v$. In other words,
we overload the notation and define $l(wv) \doteq l(v)$.
We likewise associate with each 
run $\theta$ the $\omega$-word 
$l(\theta) \doteq  l(\theta_0) l(\theta_1) \ldots$.
Suppose we are given $\varphi$,  
an LTL formula or B\"{u}chi
automaton, or any other formalism for expressing an
$\omega$-regular language over alphabet $\Sigma$. 
Let $L(\varphi) \subseteq \Sigma^\omega$ 
denote the language  expressed by $\varphi$. 
We write 
$\Pr^\sigma_u(\varphi)$ to denote the probability that
a trajectory $\theta$ of $M^\sigma_u$ satistifies $\varphi$, i.e.,  
that $l(\theta) \in L(\varphi)$.
For generality, 
rather than just allowing an initial vertex $u$
we allow an initial probability distribution $\alpha \in
{\mathcal D}(V)$.
Let $\Pr^\sigma_{\alpha}(\varphi)$ denote the probability that
under strategy $\sigma$, starting with initial distribution $\alpha$,
we will satisfy $\omega$-regular property $\varphi$.
These probabilities are well defined because the set of
such runs is  Borel measurable  (see, e.g., \cite{Var85,CY95}).

As in the introduction, 
for a $k$-tuple of $\omega$-regular properties 
$\Phi = \langle \varphi_1, \ldots, \varphi_k \rangle$,
given a strategy $\sigma$,
we let $t^\sigma = ( t^\sigma_1,\ldots, t^\sigma_k)$,
with $t^\sigma_i = \Pr^\sigma_u(\varphi_i)$, for $i=1,\ldots,k$.
For MDP $M$ and starting state $u$, 
we define the {\em achievable} set of value vectors
with respect to $\Phi$ to be 
$U_{M_u,\Phi} = \{ r \in \real_{\geq 0}^k \mid 
\mbox{ $\exists \sigma$ such that $t^\sigma \geq r$}\}$.
For a set 
$U \subseteq \real^k$, we define a subset ${\mathcal P} \subseteq U$ of it, called the
{\em Pareto curve} or the {\em Pareto set} of $U$, consisting of 
the set of 
{\em Pareto optimal} (or
 {\em Pareto efficient}) vectors inside $U$.
A vector $v \in U$ is called {\em Pareto optimal} if
$\neg \exists v'  (v' \in U  \wedge v \leq v' \wedge
v \neq v')$.
Thus
${\mathcal P} =  \{ v \in U \mid \mbox{$v$ is Pareto optimal}  \}$.
We use ${\mathcal P}_{M_u,\Phi} \subseteq U_{M_u,\Phi}$ 
to denote the Pareto curve of $U_{M_u,\Phi}$.

It is clear, e.g., from Figure \ref{fig:example_2mdp}, that the 
Pareto curve is in general an infinite set.
In fact, it follows from our results that for general $\omega$-regular
objectives the Pareto set is 
a convex polyhedral set.
In principle, we may want to compute some kind of exact representation
of this set by, e.g., 
enumerating all the vertices 
(on the upper envelope) of the 
polytope that defines the Pareto curve, or enumerating the facets 
that define it.
It is not possible to do this in polynomial-time in general.
In fact, the following theorem holds:
\begin{theorem}
There is a family of MDPs, $\langle M(n) \mid n \in \nat \rangle$, where 
$M(n)$ has $n$ states and size $O(n)$, such that for $M(n)$ 
the Pareto curve for
two reachability objectives, $\Diamond P_1$ and $\Diamond P_2$,
contains $n^{\Omega(\log n)}$ vertices (and thus $n^{\Omega(\log n)}$ facets).
\end{theorem}
\begin{proof}
We will adapt and build on a known construction for the bi-objective shortest path problem
which shows that the Pareto curve for that problem can have 
$n^{\Omega(\log n)}$ vertices.
This was shown in \cite{Carst83} 
and a simplified proof (using a similar construction) was given in
\cite{MulSha01}.
(The constructions and theorems there are phrased in terms of parametric shortest
paths, but these are equivalent to bi-objective shortest paths.)
What those constructions show is that, for some polynomial $f$, and 
for every $n$, there is a graph $G_n$ with $f(n)$ nodes and
distinguished nodes $s$ and $t$, and such that
every edge $(u,v)$  has two (positive) costs $c(u,v)$ and $d(u,v)$,
which yield two cost functions $c(\cdot)$ and $d(\cdot)$ on the $s$-$t$ paths,
such that the Pareto curve of the $s$-$t$ paths under the two objectives has 
$n^{\Omega(\log n)}$ vertices (and edges).
An important property of the constructed graphs $G_n$ is that
they are acyclic and layered,
that is, the nodes are arranged in layers $L_0={s}, L_1,L_2,\ldots,L_n={t}$,
and all edges are from layer $L_i$ to $L_{i+1}$ for some $i \in \{0,\ldots,n-1\}$.

Building on this construction, we now construct the following instance $M_n$ of the MDP problem with two 
reachability objectives. The states of $M_n$ are the same as $G_n$ with 2 extra absorbing states:
the red state $R$, and the blue state $B$, which are the two target states
of our two reachability objectives.
For each state $u$ there is one action for each outgoing edge $(u,v)$;
if we choose this action then we transition with probability $r(u,v)$ to state $R$,
with probability $b(u,v)$ to $B$, with probability $1/2$ to $v$,
and with the remaining probability to $t$.
The probabilities $r(u,v)$ and $b(u,v)$ are defined as follows.
Let $h$ be the maximum $c$ or $d$ cost over all the edges.
For an edge $(u,v)$ where $u \in L_i$ (and $v \in L_{i+1}$),
set $$r(u,v) := \frac{2^i  (2h - c(u,v))}{ 8h 2^n }$$
and $$b(u,v) := \frac{2^i  (2h - d(u,v))}{8h 2^n }$$
Note that both these quantities are in the interval $[0, 1/4]$, so all probabilities are well-defined.

The claim is that there is a 1-1 correspondence between the vertices of the
Pareto curve of this MDP $M_n$ and the Pareto curve of the bi-objective
shortest path on $G_n$.
First we note that the vertices of the Pareto curve for the MDP correspond
to pure memoryless strategies (meaning that for each vertex of the Pareto curve
a pure memoryless strategy can achieve
the value vector that the vertex defines). 
The reason for this is that the vertices are optima for a linear combination
of the two objectives, and it follows from the proof of Theorem \ref{thm:red-to-multi-LP}, which we
shall show later, that these objectives have pure memoryless optimal strategies.

A pure strategy corresponds to a path from $s$ to $t$.
Let $\pi = s, u_1, u_2, ..., u_{n-1} t$ be such a path/strategy.
The probability that this strategy leads to the red node $R$ is
$r(s,u_1) + \ldots + \Prob(\mbox{reach node $u_i$})*r(u_i,u_i+1) + \ldots$
The probability that the process reaches node $u_i$ under the strategy $\pi$
is $1/2^i$, independent of the path.
Thus, $\Prob_\pi(\mbox{reach $R$}) = a - b*c(\pi)$, where $a,b$ are constants 
independent of the
path.
Similarly, $\Prob_\pi(\mbox{reach $B$}) = a - b*d(\pi)$.

It follows that minimizing the $c$ and $d$ costs of the paths is equivalent to
maximizing the probabilities of reaching $R$ and $B$, and this also
holds for any positive linear combination of the two respective objectives.
Thus, there is
a correspondence between their Pareto curves.  
\end{proof}

So, the Pareto curve is in general a polyhedral surface of superpolynomial
size, and thus cannot be constructed exactly in polynomial time.
We show, however, 
that the Pareto set can be efficiently {\em approximated}
to any desired accuracy $\epsilon > 0$.
An {\em $\epsilon$-approximate Pareto curve},
${\mathcal P}_{M_u,\Phi}(\epsilon) \subseteq  U_{M_u,\Phi}$,
is any achievable set such that 
$\forall r \in U_{M_u,\Phi} \; \exists  t \in {\mathcal P}_{M_u,\Phi}(\epsilon)$  such that $r \leq (1 + \epsilon) t$.  
When the subscripts $M_u$ and $\Phi$ 
are clear from the context, we will drop them
and use $U$, ${\mathcal P}$, and ${\mathcal P}(\epsilon)$ to denote the achievable set, 
Pareto set, and $\epsilon$-approximate Pareto set, respectively.

We also consider general {\em multi-objective queries}.
A {\em quantitative predicate} over $\omega$-regular  property $\varphi_i$ 
is a statement 
of the form $\Pr^\sigma_u(\varphi_i) \Delta p$, 
for some rational
probability $p \in [0,1]$, and where $\Delta$ is a comparison
operator $\Delta \in \{ \leq , \geq, <, >, = \}$.
Suppose $B$ is a boolean combination over such predicates.
Then, given $M$ and $u$, and $B$, we 
can ask whether there exists a 
strategy $\sigma$ 
such that $B$ holds, or whether $B$ holds for all $\sigma$.
Note that since $B$ can be put in DNF form, and the quantification
over strategies pushed into the disjunction, 
and since $\omega$-regular languages are closed under 
complementation, any query of the form $\exists \sigma B$
(or of the form $\forall \sigma B$)
can be transformed to a disjunction (a negated disjunction, respectively) 
of queries of the form: 

\begin{equation}
\label{form:quant-query}
\exists \sigma \; \bigwedge_{i} (\mbox{Pr}^\sigma_u(\varphi_{i}) \geq r_{i} )
\wedge \bigwedge_{j} (\mbox{Pr}^\sigma_u(\psi_{j}) > r'_{j} )
\end{equation}

We call queries of the
form (1) {\em extended achievability queries}.
Thus, if the multi-objective query is fixed, it suffices to
perform a fixed number of extended achievability queries 
to decide any  multi-objective query.  
Note, however, that the number of extended achievability queries
we need could be exponential in the size of $B$.
We do not focus on optimizing query complexity in this paper.

A motivation for allowing general multi-objective queries
is to enable {\em assume-guarantee compositional reasoning} for probabilistic
systems.   Consider, e.g., a 
probabilistic system consisting of 
the concurrent composition of two components, $M_1$ and $M_2$,  
where output from $M_1$ provides input to $M_2$ and thus controls  $M_2$. 
We denote this by $M_1 \rhd M_2$. 
$M_2$ itself may generate outputs for some external device,
and $M_1$ may also be controlled 
by external inputs.
(One can also consider symmetric composition, where outputs from both 
 components provide inputs to both.  Here, for simplicity, we restrict 
ourselves to asymmetric composition where $M_1$ controls $M_2$.)
Let $M$ be an MDP with  
separate input and output action alphabets $\Sigma_1$ and $\Sigma_2$,
and let $\varphi_1$ and $\varphi_2$ denote $\omega$-regular 
properties over these two alphabets,
respectively.  We write $\langle \varphi_1 \rangle_{\geq r_1} 
M  \langle \varphi_2 \rangle_{\geq r_2}$, to denote the
assertion that {\em ``if the input controller of  $M$ satisfies $\varphi_1$
with probability $\geq r_1$, then the output generated by $M$ satisfies
$\varphi_2$ with probability $\geq r_2$''}.
Using this, we can formulate a general compositional assume-guarantee 
proof rule:

\begin{center}
$\langle \varphi_1 \rangle_{\geq r_1} M_1 \langle \varphi_2 \rangle_{\geq r_2}$\\
$\langle \varphi_2 \rangle_{\geq r_2} M_2 \langle \varphi_3 \rangle_{\geq r_3}$\\
------------------------------------\\
$\langle \varphi_1 \rangle_{\geq r_1} \; M_1 \rhd M_2 \; \langle \varphi_3 
\rangle_{\geq r_3}$
\end{center}\medskip

\noindent Thus, to check 
$\langle \varphi_1 \rangle_{\geq r_1} M_1 \rhd M_2 \langle \varphi_3
\rangle_{\geq r_3}$
it suffices to check 
two properties of smaller systems: 
$\langle \varphi_1 \rangle_{\geq r_1} M_1 \langle \varphi_2 \rangle_{\geq r_2}$
and
$\langle \varphi_2 \rangle_{\geq r_2} M_2 \langle 
\varphi_3 \rangle_{\geq r_3}$.
Note that checking $\langle \varphi_1 \rangle_{\geq r_1} M  \langle
\varphi_2 \rangle_{\geq r_2}$ amounts to checking that
there does not
exist a strategy $\sigma$ controlling $M$ such that
$\Pr^\sigma_u(\varphi_1) \geq r_1$ and $\Pr^\sigma_{u}(\varphi_2) < r_2$.

We also consider {\em qualitative multi-objective queries}.
These are queries restricted so that $B$ contains  
only {\em qualitative predicates} of the form 
$\Pr^\sigma_u(\varphi_i) \Delta b$, where $b \in \{0,1\}$.
These can, e.g., be used to check qualitative
assume-guarantee conditions of
the form: $\langle \varphi_1 \rangle_{\geq 1} M  \langle
\varphi_2 \rangle_{\geq 1}$.
It is not hard to see that again, via boolean manipulations
and complementation of automata, 
we can convert any qualitative query to a number
of queries of the  form:
\[ \exists \sigma \;  \bigwedge_{\varphi \in \Phi} 
(\mbox{Pr}^\sigma_u(\varphi) = 1) \wedge 
\bigwedge_{\psi \in \Psi}  (\mbox{Pr}^\sigma_u(\psi) > 0)\]

\noindent where $\Phi$ and $\Psi$ are sets of $\omega$-regular properties.
It thus suffices to consider only these qualitative queries.

\begin{figure*}[t]
\begin{center}
\scalebox{0.45}{\input{bad_choice.pstex_t}}
\end{center}
\caption{The MDP $M'$.\label{fig:bad_choice}}
\end{figure*}

In the next sections we study how to decide 
various classes of multi-objective
queries, and how to approximate the Pareto curve for properties $\Phi$.
Let us observe here a difficulty that
we will have to deal with.  
Namely, in general 
we will need both randomization and memory 
in our strategies in order to satisfy
even  simple qualitative multi-objective queries.
Consider the MDP, $M'$, shown in Figure \ref{fig:bad_choice}, and 
consider the conjunctive query:
$B \equiv \Pr^\sigma_u(\Box \Diamond 
P_1) > 0 \wedge \Pr^\sigma_u(\Box \Diamond P_2) > 0$.
It is not hard to see that starting at state $u$ in $M'$ 
any strategy
$\sigma$ that satisfies $B$
must use both memory and randomization.
Each predicate in $B$ can
be satisfied in isolation (in fact with probability $1$),
but, with a memoryless or deterministic strategy, if we
try to satisfy $\Box \Diamond  P_2$ with non-zero probability, 
we will be forced to satisfy $\Box \Diamond P_1$ with probability $0$.
Note, however, that we can satisfy both with probability $> 0$
using a strategy that uses both memory and randomness:
namely, upon reaching the state labeled $P_1$ for the first
time, with probability $1/2$ we use move $a$ and with probability
$1/2$ we use move $b$.  Thereafter, upon encountering the state 
labeled $P_1$ for the $n$th time, $n \geq 2$, we deterministically 
pick action $a$.
This clearly assures that both predicates are satisfied with 
probability $= 1/2 > 0$.

We note that our results (combined with 
the earlier results of \cite{CY98}) imply that for general 
multi-objective queries a randomized strategy with a {\em finite} 
amount of memory (which depends on the MDP and query) 
does suffice to satisfy any satisfiable quantitative 
multi-objective $\omega$-regular query. 

\section{Multi-objective reachability}
\label{sec:reach}

In this section, as a step towards 
quantitative multi-objective model checking problems, 
we study a simpler multi-objective reachability problem.
Specifically, we are given an MDP, $M = (V,\Gamma, \delta)$, a starting
state $u$, and
a collection of target sets $F_i \subseteq V$,
$i=1,\ldots,k$.  The sets $F_i$ may overlap.  
We have $k$ 
objectives: the $i$-th objective is to maximize the probability of 
$\Diamond F_i$, i.e., of
reaching some state in $F_i$.
We assume that the states $F = \bigcup^k_{i=1} F_i$
are all absorbing states with a self-loop.
In other words, for all $v \in F$,   
$(v,a,1,v) \in \delta$ and $\Gamma_v = \{ a \}$. (The 
assumption that target states are absorbing is necessary for the proofs in this 
section, but it is not a restriction in general for our results.  It  will
follow from the model checking results in Section 
\ref{sec:quan-model-checking}, which build on this section, 
that multi-objective
reachability problems for arbitrary target states (whether absorbing or not) 
can also be handled with the 
same complexities.)

We first need to do some preprocessing on
the MDP, to remove some useless states.
For each state $v \in V \setminus F$ we can check easily 
whether there
exists a 
strategy $\sigma$ such that
$Pr^\sigma_v(\Diamond F) > 0$: this just amounts to 
checking whether there exists a path from $v$ to $F$ in the underlying
graph of the MDP, i.e., the graph given by considering
only the non-zero-probability transitions.  Let us call a state that 
does not satisfy
this property a {\em bad} state.
Clearly, for the purposes of optimizing reachability objectives,
we can compute and remove all bad states from 
an MDP.
Thus, it is safe to assume that  
bad states do not exist.\footnote{Technically, we would
need to install a new ``dead'' absorbing state $v_{dead} \not \in F$,
such that all the probabilities going into states
that have been removed
now go to $v_{dead}$.  
For convenience in notation, instead of 
explicitly adding $v_{dead}$ we treat it as implicit: we allow that
for some states $v \in V$ and some action $a \in \Gamma_v$ we have 
$\sum_{v' \in V} p_{(v,\gamma,v')} < 1$, and we implicitly
assume that there is an ``invisible'' transition to $v_{dead}$ with the
residual probability, i.e.,  with
$p_{(v,\gamma,v_{dead})}  = 1- \sum_{v' \in V} p_{(v,\gamma,v')}$.
Of course, $v_{dead}$ would then be
a ``bad'' state, but we can ignore this implicit state.}
Let us call an MDP with goal states $F$
{\em cleaned-up} if it does not contain any bad states.

\begin{proposition}
\label{prop:clean-up}
For a cleaned-up MDP, an initial distribution $\alpha
\in {\mathcal D}(V \setminus F)$, and a vector of probabilities
$r \in [0,1]^k$, 
there exists a (memoryless) strategy $\sigma$ such that 
\[\bigwedge^k_{i=1}\mbox{\em Pr}^\sigma_\alpha(\Diamond F_i) \geq r_i\]
 if and only
if there exists a (respectively, memoryless) strategy $\sigma'$ such that
\[\bigwedge^k_{i=1}  
\mbox{Pr}^{\sigma'}_\alpha(\Diamond F_i) \geq r_i  \; \wedge
\; \bigwedge_{v \in V} \mbox{\em Pr}^{\sigma'}_v(\Diamond F) > 0\,.\] 
\end{proposition}

\begin{proof}
This is quite obvious, but we give a quick argument anyway. 
Suppose we have such a strategy $\sigma$. Since the
MDP is cleaned-up, we know that from every state in 
$V$
we can reach $F$
with a positive probability.
Suppose the strategy leads to a history whose last state is 
$v \in V \setminus F$, 
and that thereafter the strategy is such that it will never reach $F$
on any path.  We simply revise $\sigma$ 
to a strategy $\sigma'$ such that, if we ever arrive at
such a ``dead'' history, we switch
and play according to the memoryless strategy starting at $v$
which reaches $F$ with some positive probability.
Note that if $\sigma$ is memoryless then so is $\sigma'$.
\end{proof}

\begin{figure*}[t]
\begin{quote}
{\bf Objectives ($i = 1, \ldots, k$):} $\quad \quad${\bf Maximize}  $\sum_{v \in F_i} y_v$;

\vspace*{0.08in}

{\bf Subject to:}

\begin{tabbing}
$\!\sum_{\gamma \in \Gamma_v} y_{(v,\gamma)} - 
\sum_{v' \in V} \sum_{\gamma' \in \Gamma_{v'}}  p_{(v',\gamma',v)} y_{(v',\gamma')}$ \= $=$ \= $\alpha(v) \quad$ \= 
 \!\mbox{for all $v \in V\setminus F$;}\\
$\!y_v - 
\sum_{v' \in V \setminus F} \sum_{\gamma' \in \Gamma_{v'}}  p_{(v',\gamma',v)} y_{(v',\gamma')}$ \> $=$ \> $0 \quad$ \> 
 \!\mbox{for all $v \in F$;}\\
\noindent $\!y_v$ \> $\geq$ \> $0 \quad$ \> \!\mbox{for all $v \in F$;}\\
$\!y_{(v,\gamma)}$ \> $\geq$ \> $0 \quad$ \> 
\!\mbox{for all $v \in V \setminus F$ and $\gamma \in \Gamma_u$ .}
\end{tabbing}
\end{quote}
\caption{Multi-objective LP for the multi-objective MDP reachability problem
\label{fig:multi-LP}}
\end{figure*}

\noindent Now, consider the multi-objective LP described in
Figure \ref{fig:multi-LP}.\footnote{We mention 
without further elaboration that this LP can 
be derived, using complementary slackness, from the 
dual LP of the standard LP for single-objective
reachability obtained from Bellman's optimality equations,
whose variables are $x_v$, for $v \in V$, and whose
unique optimal solution is the vector $x^*$
with $x^*_v = \max_{\sigma} 
\Pr^\sigma_v(\Diamond F)$
(see, e.g., \cite{Puterman94,CY98}).}
The set of variables in this LP are as follows:
for each $v \in F$, there is a variable $y_v$, 
and for each $v \in V \setminus F$ and each 
$\gamma \in \Gamma_v$ there
is a variable $y_{(v,\gamma)}$.

\begin{theorem}
\label{thm:red-to-multi-LP}
Suppose we are given a cleaned-up MDP, $M = (V,\Gamma,\delta)$, with multiple target sets
$F_i \subseteq V$, $i = 1, \ldots, k$,  
where 
every target $v \in F = \bigcup^k_{i=1} F_i$  is an absorbing state.
Let $\alpha \in {\mathcal D}(V \setminus F)$ 
be an initial distribution (in particular 
$V \setminus F \neq \emptyset$).
Let $r \in (0,1]^k$ be a vector of positive probabilities. 
Then the following are all equivalent:
\begin{enumerate}[\em(1.)]
\item There is a (possibly randomized) memoryless strategy $\sigma$ such that 
\[\bigwedge^k_{i=1}\left(\mbox{\em Pr}^\sigma_\alpha(\Diamond F_i) \geq r_i\right) \;\]

\item There is a feasible solution $y'$ for the multi-objective
LP in Fig. \ref{fig:multi-LP} such that
\[\bigwedge^k_{i=1} \left(\sum_{v \in F_i} y'_v \geq r_i\right) \;\]

\item There is an arbitrary strategy $\sigma$ such that 
\[\bigwedge^k_{i=1}  \left(\mbox{\em Pr}^\sigma_\alpha(\Diamond F_i) \geq r_i\right) \;\]

\end{enumerate}
\end{theorem}

\begin{proof}
\mbox{}

\noindent (1.) $\Rightarrow$ (2.).
Since the MDP is cleaned up, by
Proposition \ref{prop:clean-up} we can 
assume there is a memoryless strategy $\sigma$ such that 
$\bigwedge^k_{i=1}  \Pr^\sigma_\alpha(\Diamond F_i) \geq r_i \;$ and
$\; \forall v \in V \;$ $Pr^\sigma_v(\Diamond F) > 0$.
Consider the square matrix $P^\sigma$ whose size is 
$|V \setminus F| \times |V \setminus F|$,
and whose rows and columns are indexed by states in $V \setminus F$.
The $(v,v')$'th entry  of $P^\sigma$,  $P^\sigma_{v,v'}$, is the
probability that starting in state $v$ we shall in one step end up
in state $v'$.  In other words, 
$P^\sigma_{v,v'} = \sum_{\gamma \in  \Gamma_v} \sigma(v)(\gamma) \cdot p_{v,\gamma,v'}$.

For all $v \in V \setminus F$, let 
$y'_{(v,\gamma)} =  \sum_{v' \in V \setminus F}  \alpha(v')
\sum^\infty_{n=0} (P^\sigma)^n_{v',v} \sigma(v)(\gamma) $.
In other words $y'_{(v,\gamma)}$ denotes the ``expected number of times that,
using the strategy $\sigma$, starting in the distribution $\alpha$, we
will visit the state $v$ and upon doing so choose action $\gamma$''.
We don't know yet that these are finite values, but assuming they are,
for $v \in F$, let $y'_v = 
\sum_{v' \in V \setminus F} 
\sum_{\gamma' \in \Gamma_{v'}}  p_{(v',\gamma',v)} y'_{(v',\gamma')}$.
This completes the definition of the entire vector $y'$.

\begin{lemma}
\label{lem:y-prime-good}
The vector $y'$ is well defined (i.e., all entries $y'_{(v,\gamma)}$ 
are finite).\\
Moreover, $y'$ is a feasible solution to the constraints of the
LP in Figure \ref{fig:multi-LP}.
\end{lemma}
\begin{proof}
First, we show that for all $v \in V \setminus F$
and $\gamma \in \Gamma_v$,  
$y'_{(v,\gamma)}$ 
is a well defined finite value.
It then also follows from the definition of $y'_v$
that $y'_v$ is also finite and thus that the vector $y'$ is
well defined. 
Note that 
because $\sigma$ has
the property that $\forall v \in V \;$ $Pr^\sigma_v(\Diamond F) > 0$,
$P^\sigma$ is clearly a substochastic matrix
with the property that, for some power $d \geq 1$, 
all of the row sums of $(P^\sigma)^d$ 
are strictly less than 1.
Thus, it follows that
$\lim_{n \rightarrow \infty} (P^\sigma)^n \rightarrow {\mathbf 0}$, and
thus by standard facts about matrices
the inverse matrix $(I- P^\sigma)^{-1} = \sum^\infty_{n=0} (P^\sigma)^n$
exists and is non-negative.  Now observe that
\begin{eqnarray*}
y'_{(v,\gamma)} & = & \sum_{v' \in V \setminus F}  \alpha(v')
\sum^\infty_{n=0} (P^\sigma)^n_{v',v} \sigma(v)(\gamma) \\
& = & \sum_{v' \in V \setminus F}  \alpha(v') \sigma(v)(\gamma)
\sum^\infty_{n=0} (P^\sigma)^n_{v',v} \\
& = & \sigma(v)(\gamma) \sum_{v' \in V \setminus F}  \alpha(v')  (I - P^\sigma)^{-1}_{v',v}
\end{eqnarray*}

Next, we show that $y'$ is a feasible solution to the constraints 
in the multi-objective LP
in Figure \ref{fig:multi-LP}.
Note that, for each state $v \in V \setminus F$,
the expression
$\sum_{v' \in V} 
\sum_{\gamma' \in \Gamma_{v'}}  p_{(v',\gamma',v)} y'_{(v',\gamma')}$
is precisely the ``expected number of times we will take a transition into the state $v$''
if we start at initial distribution $\alpha$ and using strategy $\sigma$,
whereas $\sum_{\gamma \in \Gamma_v} y'_{(v,\gamma)}$ defines precisely
the ``expected number of times we will take a transition out of the state $v$''.
Thus $\alpha(v)$, the probability that we will 
start in state $v$, is precisely given by $ 
\sum_{\gamma \in \Gamma_v} y'_{(v,\gamma)} - 
\sum_{v' \in V} 
\sum_{\gamma' \in \Gamma_{v'}}  p_{(v',\gamma',v)} y'_{(v',\gamma')} = \alpha(v)$.
More formally, for each state $v \in V \setminus F$:
\begin{eqnarray*}
\sum_{v' \in V} \sum_{\gamma' \in \Gamma_{v'}}  p_{(v',\gamma',v)} y'_{(v',\gamma')}
& = & \sum_{v' \in V} 
\sum_{\gamma' \in \Gamma_{v'}}  p_{(v',\gamma',v)} 
\sum_{v'' \in V \setminus F}  \alpha(v'')
\sum^\infty_{n=0} (P^\sigma)^n_{v'',v'} \sigma(v')(\gamma')\\
& = & 
\sum_{v'' \in V \setminus F}  \alpha(v'')
\sum_{v' \in V} 
\sum_{\gamma' \in \Gamma_{v'}}  p_{(v',\gamma',v)} 
\sum^\infty_{n=0} (P^\sigma)^n_{v'',v'} \sigma(v')(\gamma')\\
& = & 
\sum_{v'' \in V \setminus F} \alpha(v'') \sum^\infty_{n=1} (P^\sigma)^n_{v'',v}
\end{eqnarray*}
The last expression is easily seen to be the expected number of times we will 
transition into state $v$.
It is clear by linearity of expectations that 
$\sum_{\gamma \in \Gamma_v} y'_{(v,\gamma)}$ gives the expected number of times
we will transition out of state $v$.  It is thus clear that
$\sum_{\gamma \in \Gamma_v} y'_{(v,\gamma)} - 
\sum_{v' \in V} 
\sum_{\gamma' \in \Gamma_{v'}}  p_{(v',\gamma',v)} y'_{(v',\gamma')} = \alpha(v)$.   
\end{proof}

Now we argue that $\sum_{v \in F_i}  y'_v  = \Pr^\sigma_\alpha(\Diamond F_i)$. 
To see this, note that for $v \in F$, $y'_v = \sum_{v' \in V \setminus F} 
\sum_{\gamma' \in \Gamma_{v'}}  p_{(v',\gamma',v)} y'_{(v',\gamma')}$ 
is precisely the ``{\em expected number of times that we will transition into 
state $v$
for the first time}'', starting at distribution $\alpha$.  
The reason we can say ``for the first time'' is because
only the states in $V \setminus F$ are included in the matrix $P^\sigma$. 
But note that this italicised statement in quotes 
is another way to define the probability of eventually reaching state $v$.
This equality can be establish formally, but we omit the formal algebraic 
derivation here.
Thus $\sum_{v \in F_i} y'_v = \Pr^\sigma_\alpha(\Diamond F_i) \geq r_i$.
We are done with (1.) $\Rightarrow$ (2.).\\

\noindent (2.) $\Rightarrow$ (1.).   We now wish to show that if $y''$ is
a feasible solution to the multi-objective LP such that 
$\sum_{v \in F_i} y''_v  \geq r_i > 0$, for all $i=1,\ldots,k$, 
then there exists a memoryless
strategy $\sigma$ such that $\bigwedge^k_{i=1}\Pr^\sigma_\alpha(\Diamond F_i) \geq r_i$.

Suppose we have such a solution $y''$.  
Let $S = \{v \in V \setminus F \mid 
\sum_{\gamma \in \Gamma_v} y''_{(v,\gamma)} > 0 \}$.
Let $\sigma$ be the memoryless strategy, given as follows.
For each $v \in S$ 
$$ \sigma(v)(\gamma) := \frac{y''_{(v,\gamma)}} { \sum_{\gamma' \in \Gamma_v} y''_{v,\gamma'}}$$
Note that since $\sum_{\gamma \in \Gamma_v} y''_{(v,\gamma)} > 0$, 
$\sigma(v)$ is a well-defined probability distribution on the moves at 
state $v \in S$.
For the remaining states $v \in (V \setminus F) \setminus S$,
let $\sigma(v)$ be an arbitrary distribution in ${\mathcal D}(\Gamma_v)$.

\begin{lemma}
\label{lem:sigma-good}
This memoryless strategy $\sigma$ satisfies 
$\bigwedge^k_{i=1} Pr^\sigma_\alpha(\Diamond F_i) \geq r_i$.
\end{lemma}
\begin{proof} 
Let us assume, for the sake of convenience in our
analysis, that there is an extra  
dead-end absorbing state $v_{dead}
\not\in F$ available,  and an extra move $\gamma_{dead}$ 
available at each state, $v$, with $p_{(v,\gamma_{dead},v_{dead})} = 1$,
and for each $v \in (V \setminus F) \setminus S$,
instead of letting $\sigma(v)$ be arbitrary,  
let $\sigma(v)(\gamma_{dead}) = 1$.
In other words, from each such state we simply move directly to an
absorbing dead-end which is outside of $F$.  
The assumption that such a dead-end exists is
just for convenience: clearly, without such a dead-end, 
we can use any (mixed) move at such vertices in our strategy,
and such a strategy 
would yield at least as high a
value for $\Pr^\sigma_{\alpha}(\Diamond F_i)$, for all $i=1,\ldots,k$.

Let us now explain the reason why we don't care about what moves are
used at states outside $S$ in the strategy $\sigma$.
Let $\support(\alpha) = \{ v \in V \setminus F \mid \alpha(v) > 0\}$.
We claim $S$ contains all states reachable from $\support(\alpha)$
using strategy $\sigma$.
To see this,   
first note that $\support(\alpha) \subseteq S$, because
for all $v \in \support(\alpha)$, 
since 
$\sum_{\gamma \in \Gamma_v} y''_{(v,\gamma)} - 
\sum_{v' \in V} \sum_{\gamma' \in \Gamma_{v'}}  p_{(v',\gamma',v)} 
y''_{(v',\gamma')} = 
\alpha(v)$ and $\alpha(v) > 0$, and since $y''_{v',\gamma'} \geq 0$ for all 
$v' \in V \setminus F$ and 
$\gamma' \in \Gamma_{v'}$,  
it must be the case that, 
$\sum_{\gamma \in \Gamma_v} y''_{(v,\gamma)} > 0$.  
Thus $\support(\alpha) \subseteq S$.
Inductively, for $k \geq 0$, consider any state $v \in V \setminus F$, 
such that we can, with non-zero probability, 
reach $v$ in $k$ steps using strategy $\sigma$ from a state in $\support(\alpha)$,
and such that we can not reach $v$ 
(with non-zero probability) in any fewer than $k$ step.
For the base case $k=0$, we already know $v \in \support(\alpha) \subseteq S$.
For $k> 0$, we must have $\alpha(v) = 0$. But note that 
there must be a positive probability of 
moving to $v$ in one step from some other
state $v'$ which can be reached in $k-1$ steps from $\support(\alpha)$.
But this is so if and only if for some $\gamma' \in \Gamma_{v'}$,
both $p_{(v',\gamma',v)} > 0$ and $y''_{(v',\gamma')} > 0$ (
and thus $\sigma(v')(\gamma') > 0$).
Hence, 
$\sum_{v' \in V} \sum_{\gamma' \in \Gamma_{v'}}  p_{(v',\gamma',v)} 
y''_{(v',\gamma')} > 0$.
Thus since 
$\sum_{\gamma \in \Gamma_v} y''_{(v,\gamma)} - 
\sum_{v' \in V} \sum_{\gamma' \in \Gamma_{v'}}  p_{(v',\gamma',v)} 
y''_{(v',\gamma')} = 
0$, we must have $\sum_{\gamma \in \Gamma_v} y''_{(v,\gamma)} > 0$,
and thus $v \in S$.
Hence $S$ contains the set of nodes reachable from nodes in the 
support of the initial
distribution, 
$\support(\alpha)$, using the strategy $\sigma$.

We will now show that  $\Pr^\sigma_\alpha(\Diamond F_i) \geq r_i$, for 
all $i = 1,\ldots,k$.
Let us consider the underlying graph of the ``flows'' defined by $y''$.
Namely, let $G =(V,E)$ be a graph on states of $M$ such that $(v,v') \in E$
if and only if there is some $\gamma \in \Gamma_v$ such that
$y''_{(v,\gamma)} > 0$  and $p_{(v,\gamma,v')} > 0$.
Let $W \subseteq V \setminus F$ be the set of vertices in $V \setminus F$
that have a non-zero ``flow'' to $F$, i.e.,
$v$ is in $W$  iff there is a path in $G$ from $v$ to some vertex in $F$.

For $v \in V \setminus F$, let $z_v = \sum_{\gamma \in \Gamma_v} y''_{(v,\gamma)}$.
Note that by the constraints of the LP, for any vertex $v \in S$
\begin{eqnarray*}
\alpha(v)  & = & z_v  - \sum_{v' \in V \setminus F} \sum_{\gamma \in \Gamma_{v'}} 
p_{(v',\gamma,v)} y''_{(v',\gamma)}\\
& = & z_v  - \sum_{v' \in S} \sum_{\gamma \in \Gamma_{v'}} 
p_{(v',\gamma,v)} y''_{(v',\gamma)}  \quad \mbox{(because all flow into $v$ 
comes from $S$)}\\
& = &  z_v - \sum_{v' \in S} \sum_{\gamma \in \Gamma_{v'}} 
p_{(v',\gamma,v)}  y''_{(v',\gamma)} \frac{z_{v'}}{\sum_{\gamma' \in \Gamma_{v'}}
y''_{(v',\gamma')}}\\
& = & z_v - \sum_{v' \in S} \sum_{\gamma \in \Gamma_{v'}} 
p_{(v',\gamma,v)} \sigma(v')(\gamma) z_{v'}\\
& = & z_v - \sum_{v' \in S}  P^\sigma_{v',v} z_{v'}
\end{eqnarray*}
Now, let us focus on the vertices in $W$.   
Note that, by definition, $W \subseteq S$. 
Consider the submatrix $P^\sigma_{W,W}$
obtained from $P^\sigma$ by eliminating the rows and columns whose indices are
not in $W$.
Note that since there is no flow into a vertex in $W$ from 
a vertex outside of $W$, the above equalities yield,
for each $v \in W$, 
$\alpha(v) =  z_v - \sum_{v \in W}  P^\sigma_{v',v} z_{v'}$.
This can be written in matrix notation as  $\alpha^T|_W = z^T|_W (I - P^\sigma_{W,W})$.

Now, note that since every vertex in $W$ has a ``flow'' to $F$, in terms of the
underlying Markov chain of the substochastic matrix $P^\sigma_{W,W}$, this means
that every vertex in $W$ is transient, and that there is a power $d \geq 1$, such that
$(P^\sigma_{W,W})^d$ has the property that all its row sums are strictly less than $1$.
Consequently,  $\lim_{d \rightarrow \infty} (P^\sigma_{W,W})^d = {\mathbf 0}$ and
the matrix $(I-P^\sigma_{W,W})$ is invertible, with $(I - P^\sigma_{W,W})^{-1}
= \sum^\infty_{i=0} (P^\sigma_{W,W})^i)$, a nonnegative matrix.
Thus, 
$z^T|_W 
=  \alpha^T (I - P^\sigma_{W,W})^{-1} = \alpha^T (\sum^\infty_{i=0} (P^\sigma_{W,W})^i$.
From this it follows, again because no vertex outside of $W$ has a flow into $W$, 
that for each $v \in W$:
\begin{eqnarray*}
z_v & = & \sum_{v' \in V \setminus F}  \alpha(v')
\sum^\infty_{n=0} \sum_{\gamma \in \Gamma_{v}} (P^\sigma)^n_{v',v} \sigma(v)(\gamma)\\
& = &
\sum_{\gamma \in \Gamma_{v}}  \sum_{v' \in W}   \alpha(v')
\sum^\infty_{n=0} (P^\sigma)^n_{v',v} \sigma(v)(\gamma)\\ 
& & \quad \quad
\mbox{(because
all moves into $v$ of strategy $\sigma$ come from vertices in $W$)}\\
& = &  \sum_{\gamma \in \Gamma_v} y'_{(v,\gamma)}
\end{eqnarray*}
where, in the last expression, the values $y'_{(v,\gamma)}$, not to be mistaken with 
$y''_{(v,\gamma)}$, are values from the vector $y'$ which we obtained 
in the proof that $(1.) \Rightarrow (2.)$,
from 
a given memoryless strategy
$\sigma$.  In this case, the strategy $\sigma$ in question is precisely the memoryless
strategy we just defined based on $y''$.
Thus, for all $v \in W$:
\begin{equation}
z_v = \sum_{\gamma \in \Gamma_v} y''_{(v,\gamma)} = \sum_{\gamma \in \Gamma_v} y'_{(v,\gamma)}
\label{eqn:tot_same}
\end{equation}
We next show that in fact for all $v \in W$ and $\gamma \in \Gamma_v$, 
$y''_{(v,\gamma)} = y'_{(v,\gamma)}$.  For $v \in W$ and $\gamma \in \Gamma_v$,
we have:
\begin{eqnarray*}
y'_{(v,\gamma)} & = & \sum_{v' \in W}  \alpha(v') \sum^\infty_{i=0} (P^\sigma_{W,W})^i_{v',v} 
\sigma(v)(\gamma)\\
& = & \sum_{v' \in W} \alpha(v')  \sum^\infty_{i=0} (P^\sigma_{W,W})^i_{v',v} 
\frac{y''_{(v,\gamma)}}{\sum_{\gamma' \in \Gamma_v} y''_{(v,\gamma')}}\\
& = & \frac{y''_{(v,\gamma)}}{\sum_{\gamma' \in \Gamma_v} y''_{(v,\gamma')}} 
\sum_{v' \in W} \alpha(v')  \sum^\infty_{i=0} (P^\sigma_{W,W})^i_{v',v}
\end{eqnarray*}
But recall that the ``expected number of times we will transition out of state $v$''
is given by
$\sum_{\gamma \in \Gamma_v} y'_{(v,\gamma)} = \sum_{v' \in W} \alpha(v')  \sum^\infty_{i=0} (P^\sigma_{W,W})^i_{v',v}$.

Hence $y'_{(v,\gamma)} = \frac{y''_{(v,\gamma)}}{\sum_{\gamma' \in \Gamma_v} y''_{(v,\gamma')}}
\sum_{\gamma \in \Gamma_v} y'_{(v,\gamma)}$.
Thus, by using equation (\ref{eqn:tot_same}) and canceling, we get 
$y'_{(v,\gamma)} = y''_{(v,\gamma)}$.
Thus, since $y''$ is a feasible solution to the LP, we have that for any $v \in F$:
\begin{eqnarray*}
y''_v & = &  \sum_{v' \in V \setminus F} \sum_{\gamma' \in \Gamma_{v'}} p_{(v',\gamma',v)}
y''_{(v', \gamma')}\\
& = & \sum_{v' \in W} \sum_{\gamma' \in \Gamma_{v'}} p_{(v',\gamma',v)}
y''_{(v',\gamma')} \quad \mbox{(because all flow into $F$ is from $W$)}\\
& = & \sum_{v' \in W} \sum_{\gamma' \in \Gamma_{v'}} p_{(v',\gamma',v)}
y'_{(v',\gamma')}\\
& = & \mbox{Pr}^\sigma_\alpha(\Diamond \{ v \})
\end{eqnarray*}
The last equality holds because, as we showed in the proof of ((1.) $\Rightarrow$ (2.)),
the expression $\sum_{v' \in W} \sum_{\gamma' \in \Gamma_{v'}} p_{(v',\gamma',v)}
y'_{(v',\gamma')} = \sum_{v' \in V \setminus F} \sum_{\gamma' \in \Gamma_{v'}} p_{(v',\gamma',v)}
y'_{(v',\gamma')}$ is 
exactly the ``{\em expected number of times that we will visit the vertex $v \in F$ for
the first time}'', which is precisely the probability $\Pr^\sigma_\alpha(\Diamond \{ v \})$.

Thus, clearly, $\sum_{v \in F_i} y''_{v} = \sum_{v \in F_i} \Pr^\sigma_\alpha(\Diamond \{ v \})
= \Pr^\sigma_\alpha(\Diamond F_i)$.
Thus, since we have assumed that $\sum_{v \in F_i} y''_{v} \geq r_i$,
we have established that $\Pr^\sigma_\alpha(\Diamond F_i) \geq r_i$, for all target sets $F_i$.
\end{proof}

This completes the proof that (2.) $\Rightarrow$ (1.).

\vspace*{0.09in}

\noindent  
(3.) $\Leftrightarrow$ (1.).  Clearly (1.) $\Rightarrow$ (3.),
so we need to show that (3.) $\Rightarrow$ (1.).

Let $U$ be the set of achievable vectors, i.e.,
all $k$-vectors $r = \langle r_1 \ldots r_k \rangle$ 
such that there is a (unrestricted) strategy $\sigma$ 
such that $\bigwedge^k_{i=1} \Pr^\sigma_\alpha(\Diamond F_i) \geq r_i$. 
Let $U^{\odot}$ be the analogous set where the strategy $\sigma$ is restricted
to be a possibly randomized but memoryless (stationary) strategy.
Clearly, $U$ and $U^{\odot}$ are both downward closed, i.e.,
if $r \geq r'$ and $r \in U$ then also 
$r' \in U$, and similarly with $U^{\odot}$. 
Also, obviously $U^{\odot} \subseteq U$.
We characterized $U^{\odot}$ in (1.) $\Leftrightarrow$ (2.), 
in terms of a multi-objective LP.
Thus, $U^{\odot}$ is the projection of the feasible space of a set
of linear inequalities (a polyhedral set), namely the
set of inequalities in the variables $y$ given in Fig. \ref{fig:multi-LP} 
and the inequalities $\sum_{v \in F_i} y_v \geq r_i$, $i=1,\ldots,k$.
The feasible space is a polyhedron in the space indexed by the
$y$ variables and the $r_i$'s, and $U^{\odot}$ is its projection on
the subspace indexed by the $r_i$'s.
Since the projection of a convex set is convex, it follows that
$U^{\odot}$ is convex.

Suppose that there is a point $r \in U \setminus U^{\odot}$.
Since $U^{\odot}$ is convex, this implies that there is a separating hyperplane
(see, e.g., \cite{GLS93})
that separates $r$ from $U^{\odot}$, and in fact since $U^{\odot}$ is downward closed,
there is a separating hyperplane with non-negative coefficients,
i.e. there is a non-negative ``weight'' vector 
$w= \langle w_1, \ldots , w_k \rangle$ such that
$w^T r = \sum^k_{i=1} w_i r_i > w^Tx$ for every point $x \in U^{\odot}$.

Consider now the MDP $M$ with the following undiscounted reward structure.
There is $0$ reward for every state, action and transition,
except for transitions to a state $v \in F$ from a state
in $V \setminus F$; i.e. 
a reward is produced only once, in the first transition
into a state of $F$.
The reward for every transition to a state $v \in F$
is $\sum \; \{w_i \mid i \in \{1,\ldots,k\} \; \mbox{\small \&} \;  v \in F_i\}$.
By the definition, the expected reward of a policy $\sigma$
is $\sum_{i=1}^k w_i \Pr^{\sigma}_\alpha (\Diamond F_i)$. 
From classical MDP theory, we know that there is a memoryless strategy
(in fact even a deterministic one) that maximizes the expected
reward for this type of reward structure. 
(Namely, this is a positive bounded reward case: 
 see, e.g., Theorem 7.2.11 in \cite{Puterman94}.)
Therefore, $\max \{ w^T x \mid x \in U \} = \max \{ w^Tx \mid x \in U^{\odot} \}$,
contradicting our assumption that $w^Tr >  \max \{ w^Tx \mid x \in U^{\odot} \}$.
\end{proof}

\begin{corollary}
\label{cor:ach-and-pareto}
Given an MDP $M = (V,\Gamma,\delta)$,  a number of target sets
$F_i \subseteq V$, $i = 1, \ldots, k + k'$, such that every
state $v \in F = \bigcup^{k+k'}_{i=1} F_i$  is absorbing,
and 
an initial state 
$u$ (or even initial distribution $\alpha \in {\mathcal D}(V)$):
\begin{enumerate}[\em(a.)]

\item Given an extended achievability query for reachability,
$\exists \sigma B$, where \[B \equiv 
\bigwedge^k_{i=1} (\mbox{Pr}^\sigma_u(\diamond F_{i}) \geq r_{i} )
\wedge \bigwedge^{k+k'}_{j={k+1}} (\mbox{Pr}^\sigma_u(\Diamond F_{j}) 
> r_{j} ),\]

\noindent we can in time polynomial in the size of the input, 
$|M|+ |B|$,  
decide whether $\exists \sigma \; B$ is satisfiable and  
if so construct a memoryless strategy 
that satisfies it.

\item For $\epsilon > 0$, we 
can compute an $\epsilon$-approximate Pareto curve ${\mathcal P}(\epsilon)$ 
for the multi-objective reachability problem with objectives
$\Diamond F_i$, $i=1,\ldots,k$,
in time polynomial in $| M |$ and $1/\epsilon$.
\end{enumerate}
\end{corollary}

\begin{proof}
For (a.),  consider the constraints of the LP in Figure \ref{fig:multi-LP},
and add the following constraints:
for each $i \in \{1,\ldots,k\}$ add the constraint
$\sum_{v \in F_i} y_{v} \geq r_i$, and for each $j \in \{k+1,\ldots,k+k'\}$,
add the constraint $\sum_{v \in F_j} y_{v} \geq r_j + z$,  where
$z$ is a new variable, and also add the constraint $z \geq 0$.
Finally, consider the new objective ``Maximize  $z$''.
Solve this LP to find whether an optimal feasible solution $y^*,z^*$ exists,
and if so whether $z^*>0$.
If no solution exists, or if $z^* \leq 0$, then the extended achievability
query is not satisfiable.
Otherwise, if $z^* > 0$, then a strategy that satisfies $\exists \sigma B$ 
exists,
and moreover we can construct a memoryless strategy that
satisfies it by using the vector $y'' = y^*$ 
and
picking the strategy $\sigma$ constructed from $y''$ in the proof of (2.) 
$\Rightarrow$ (1.) in Theorem \ref{thm:red-to-multi-LP}.

Part (b.) is immediate from Theorem \ref{thm:red-to-multi-LP},
and the results of \cite{PapaYan00}, which show
we can $\epsilon$-approximate the Pareto curve for multi-objective
linear programs in time polynomial in
the size of the constraints and objectives and in $1/\epsilon$.
\end{proof}

\section{Qualitative multi-objective model checking}
\label{sec:qual}

\begin{theorem}
\label{thm:qual-multi}
Given an MDP $M$, an initial state $u$, and
a qualitative multi-objective query $B$,
we can decide whether there exists a strategy $\sigma$ 
that satisfies $B$, and if so construct such a strategy,
in time polynomial 
in $|M|$, and using only graph-theoretic methods
(in particular, without linear programming).
\end{theorem} 

\begin{proof} 
By the discussion in Section \ref{sec:basics}, 
it suffices to consider the case where
we are given an MDP, $M$, and two sets of 
$\omega$-regular properties $\Phi, \Psi$, and we want a
strategy $\sigma$ such that
\[  \bigwedge_{\varphi \in \Phi} \mbox{Pr}^\sigma_u(\varphi) = 1 \wedge 
\bigwedge_{\psi \in \Psi}  \mbox{Pr}^\sigma_u(\psi) > 0\]

\vspace*{-0.04in}

\noindent Assume the properties in $\Phi$, $\Psi$ are all given by 
(nondeterministic) B\"{u}chi automata $A_i$.
We will use and build on results 
in \cite{CY98}.  In \cite{CY98} (Lemma 4.4, page 1411) 
it is shown that we can construct 
from $M$ and from a collection $A_i$, $i=1,\ldots,m$, of
B\"{u}chi automata,
a new MDP $M'$ (a refinement of $M$)
which is the ``product'' of 
$M$ with the {\em naive determinization}
of all the $A_i$'s 
(i.e., the result of applying the standard subset construction
on each $A_i$, without imposing any acceptance 
condition).  Technically, we have to slightly 
adapt the constructions 
of \cite{CY98}, which use the convention that MDP states 
are either purely controlled or purely probabilistic,
to the convention used in this paper which combines both
control and probabilistic behavior at each state.  But these 
adaptations are straightforward.
For completeness, we recall the (adapted) formal definition of $M'$. 
The states of the MDP $M'$ are tuples $(x,z_1,\ldots,z_m)$, where
$x$ is a state of the MDP, $M$, and $z_i$ is a {\em set of states} of
$A_i$.  The transition relation $\delta'$ of $M'$  is as follows.
There exists a transition $((x,z_1,...,z_m),a,p,(x',z'_1,\ldots,z'_m))
\in \delta'$
if and only if the transition $(x,a,p,x')$ is in $M$
and, for each $i=1,\ldots,m$, $z'_i$  
is precisely the set of states in the B\"{u}chi automaton $A_i$ 
that one could
reach with one transition, starting from some state in the set $z_i$
and reading the symbol $l(x')$.
Technically, we also have to add a dummy initial state $x_0$ to the MDP, $M$,
such that there is a single enabled action, $\gamma_0$, at $x_0$,
and such that there are 
transitions from $x_0$ on action $\gamma_0$ 
to other states according to some initial
probability distribution on states, $\alpha \in {\mathcal D}(V)$.
Thus, in particular, if we assume 
there is just one initial state $u$ in the MDP, $M$, 
then we would now have one transition $(x_0,\gamma_0,1,u) \in \delta$ 
in the new $M$
with added dummy state $x_0$.
The reason for adding the dummy $x_0$ is because
our definition of the product $M'$ 
does not use the label of the initial state in defining the
transitions of $M'$.
We also assume, w.l.o.g., that each B\"{u}chi automaton $A_i$  
has a single initial state $s_0^i$.
In this way, the initial state of $M'$ becomes the tuple 
$v_0 = (x_0,\{s_0^1\},\ldots,\{s_0^m\})$.

By Lemma 4.4 and 4.5 of \cite{CY98}, this MDP $M'$ has the following 
two properties.
For every subset $R$ of $\Phi \cup \Psi$ there is a subset 
$T_R$ of corresponding ``target  states''
of $M'$ (and we can compute this subset efficiently, in time
polynomial in the size of $M'$) that 
satisfies the following two conditions:

\begin{enumerate}[(a)]
\item[(I)] If a trajectory of $M'$ hits a state in $T_R$ at some point,  
then we can apply from that point on a strategy $\mu_R$ (which is
deterministic but uses memory)
which ensures that the resulting infinite trajectory satisfies all properties
in $R$ almost surely (i.e., with conditional probability 1, conditioned on the initial prefix that hits $T_R$). 
\item[(II)] For every strategy, the set of trajectories that satisfy all properties in $R$ 
and do not infinitely often hit some state of $T_R$ has probability 0.
\end{enumerate}
We now outline the algorithm for deciding qualitative multi-objective
queries.
\begin{enumerate}[(1)]
\item Construct the MDP $M'$ from $M$ and from the properties $\Phi$ and 
$\Psi$ (in other words, using one automaton for each property in $\Phi$ 
and one for each property in $\Psi$).

\item Compute $T_\Phi$, and compute for each property $\psi_i \in \Psi$ 
the set of states $T_{R_i}$ where $R_i = \Phi \cup \{\psi_i\}$.\footnote{
\label{foot:actually}Actually these sets can all be computed together: we can compute
maximal {\em closed}
components of the MDP, determine the properties that each component {\em 
favors} 
(see Def. 4.1 of \cite{CY98}), and tag each state with the sets for which it is a target state.}

\item If $\Phi \neq \emptyset$, prune $M'$ by 
identifying and removing all ``bad'' states by applying the following rules. 
\begin{enumerate}
\item All states $v$ that cannot ``reach''
any state in $T_\Phi$ are ``bad''.\footnote{By ``reach'', we
mean that starting at the state $v = v_0$, there a sequence of transitions
$(v_i,\gamma,p_i,v_{i+1}) \in \delta$, $p_i > 0$, such that $v_n \in T_{\Phi}$ for some 
$n \geq 0$.} 
\item If for a state $v$  
there is an action $\gamma \in \Gamma_v$
such that there is a transition $(v,\gamma,p,v') \in \delta'$, $p >0$,
and $v'$ is bad, then remove $\gamma$ from $\Gamma_v$.
\item If for some state $v$, $\Gamma_v = \emptyset$, then mark $v$ as bad. 
\end{enumerate}
Keep applying these rules until no more states can be labelled bad and no more
actions removed for any state.

\item Restrict $M'$ to the reachable states (from the initial state $v_0$) 
that are not bad, and restrict their action sets to actions 
that have not been removed, 
and let $M''$ be the resulting MDP.
\item
If ($M''=\emptyset$ or  $\exists \psi_i \in \Psi$ such that $M''$
does not contain any state of $T_{R_i}$ )\\ \hspace*{0.1in} then return No. \\
Else return Yes.
\end{enumerate}

\noindent {\em Correctness proof:}
In one direction, suppose there is a strategy $\sigma$ such that
$ \bigwedge_{\varphi \in \Phi} \mbox{Pr}^\sigma_u(\varphi) = 1 \wedge 
\bigwedge_{\psi \in \Psi}  \mbox{Pr}^\sigma_u(\psi) > 0$.
First, note that there cannot be any finite 
prefix of a trajectory under $\sigma$ that hits a state
that cannot reach any state in $T_\Phi$.
For, if there was such a path, then all trajectories that start with this 
prefix would
go only finitely often through $T_\Phi$. Hence (by property (II) above)
almost all these trajectories do not satisfy all properties in $\Phi$,
which contradicts the fact that 
all these properties have probability 1 under $\sigma$.
From the fact that no path under $\sigma$ hits a state that cannot reach 
$T_\Phi$,
it follows by an easy induction that no finite 
trajectory under $\sigma$ hits any bad 
state.
That is, under $\sigma$ all trajectories stay in the sub-MDP $M''$.
Since every property $\psi_i \in \Psi$ has probability 
$\mbox{Pr}^\sigma_u(\psi_i)>0$ and
almost all trajectories that satisfy $\psi_i$ and $\Phi$ must hit a state of 
$T_{R_i}$ 
(property (II) above),
it follows that $M''$ contains some state of $T_{R_i}$ for each 
$\psi_i \in \Psi$.  Thus the algorithm returns Yes.

In the other direction, suppose that the algorithm returns Yes. 
First, note that for all states $v$
of $M''$, and all enabled actions $\gamma \in \Gamma_v$ in $M''$,
all transitions $(v,\gamma,p,v') \in \delta$, $p>0$ of  $M'$
must still be in $M''$ (otherwise, $\gamma$ would have been
removed from $\Gamma_v$ at some stage using rule 3(b)).
On the other hand, some states may have some missing actions in $M''$.
Next, note that all bottom strongly connected components (BSCCs) of $M''$
(to be more precise, in the underlying one-step reachability graph of $M''$)
contain a state of $T_\Phi$ (if $\Phi = \emptyset$ then all states are 
in $T_\Phi$), for otherwise the states in these BSCCs 
would have been eliminated
at some stage using rule 3(a). 

Define the following strategy $\sigma$ which works in two phases.
In the first phase, the trajectory stays within $M''$.
At each control state take a random action that remains in $M''$ out of the
state; the probabilities do not matter, we can use any non-zero 
probability for all the remaining actions.
In addition, at each state, if the state is in $T_\Phi$ or it is in
$T_{R_i}$ for some property $\psi_i \in \Psi$, then with some nonzero
probability the strategy decides to terminate phase 1 and move to
phase 2 by switching to the strategy $\mu_\Phi$ or $\mu_{R_i}$ respectively,
which it applies from that point on.
(Note: a state may belong to several $T_{R_i}$'s, in which case each one
of them gets some non-zero probability - the precise value is unimportant.)
 
We claim that this strategy $\sigma$ meets the desired requirements -
it ensures probability 1 for all properties in $\Phi$ and 
positive probability for all properties in $\Psi$.
For each $\psi_i \in \Psi$, the MDP $M''$ contains some state of $T_{R_i}$;
with nonzero probability the process will follow a path to that state and
then switch to the strategy $\mu_{R_i}$ from that point on, 
in which case it will satisfy $\psi_i$ (property (I) above). 
Thus, all properties in $\Psi$ are satisfied with positive probability.

As for $\Phi$ (if $\Phi \neq \emptyset$), note that with probability 1
the process will switch at some point to phase 2, because all
BSCCs of $M''$ have a state in $T_\Phi$.
When it switches to phase 2 it applies strategy $\mu_\Phi$ or
$\mu_{R_i}$ for some $R_i = \Phi \cup \{\psi_i\}$, hence in either 
case it will
satisfy all properties of $\Phi$ with probability 1.
\end{proof}

\section{Quantitative multi-objective model checking.}

\label{sec:quan-model-checking}

\begin{theorem}
\label{thm:quan-multi}
\mbox{}
\begin{enumerate}[\em(1.)]
\item Given an MDP $M$, an initial state $u$, 
and a quantitative multi-objective query $B$,
we can decide whether there exists a strategy $\sigma$ 
that satisfies $B$, and if so construct such a strategy,
in time polynomial 
in $|M|$. 

\item Moreover, given $\omega$-regular properties 
$\Phi = \langle \varphi_1,\ldots,\varphi_k \rangle$,
we can construct an $\epsilon$-approxi\-mate Pareto curve 
$P_{M_u,\Phi}(\epsilon)$, for the set of achievable
probability vectors $U_{M_u,\Phi}$
in time polynomial in $M$ and in $1/\epsilon$.
\end{enumerate}
\end{theorem}

\begin{proof} 
For (1.), by the discussion in Section \ref{sec:basics}, we 
only need to  consider extended achievability queries,
$B \equiv \bigwedge^{k'}_{i=1} \Pr^\sigma_u(\varphi_i) \geq r_i
\wedge \bigwedge^{k}_{j=k'+1} \Pr^\sigma_u(\varphi_j) > r_j $,
where $k \geq k' \geq 0$, and
for a vector $r \in (0,1]^{k}$.
Let $\Phi = \langle \varphi_1,\ldots,\varphi_{k} \rangle$.
We are going to reduce this multi-objective problem with objectives
$\Phi$ to the
quantitative multi-objective reachability
problem studied in Section \ref{sec:reach}. 
From our reduction, both (1.) and  (2.) will follow, using
Corollary \ref{cor:ach-and-pareto}.
As in the proof of Theorem \ref{thm:qual-multi}, we will
build on constructions from \cite{CY98}:   
form the MDP $M'$ consisting of the product of $M$ with the
naive determinizations of the automata $A_i$ for the properties 
$\varphi_i \in \Phi$.
For each subset $R \subseteq \Phi$ we determine the corresponding subset 
$T_R$ of target states in $M'$.\footnote{Again, we don't need to compute these 
sets separately.  See Footnote \ref{foot:actually}.}

Construct the following MDP $M''$. 
Add to $M'$ a new absorbing state $s_R$ for each subset $R$ of $\Phi$.
For each state $u$ of $M'$ and each maximal subset $R$
such that $u \in T_R$ add a new action $\gamma_{R}$ to $\Gamma_u$,
and a new transition $(u,\gamma_{R},1,s_R)$ to $\delta$.
With each property $\varphi_i \in \Phi$ we associate the subset
of states $F_i =\{ s_R \mid \varphi_i \in R \}$.
Let $\overline{F} = \langle \Diamond F_1, \ldots, \Diamond F_k \rangle$.
Let $u^*$ be the initial state of the product MDP $M''$, given by 
the start state $u$ of $M$ and the start states of all the 
naively determinized $A_i$'s.
Recall that $U_{M_u,\Phi} \subseteq [0,1]^k$ 
denotes the achievable set
for the properties $\Phi$ in $M$ starting at $u$, and
that $U_{M''_{u^*},\overline{F}}$ denotes the achievable set for $\overline{F}$
in $M''$ starting at $u^*$.
\begin{lemma}
\label{lem:quant-ach-same}
$U_{M_u, \Phi} = U_{M''_{u^*},\overline{F}}$.
Moreover, from a strategy $\sigma$ that achieves $r$ in $U_{M_u,\Phi}$,
we can recover a strategy $\sigma'$ that achieves $r$ in 
$U_{M''_{u^*},\overline{F}}$, and vice versa.
\end{lemma}
\begin{proof}
One direction is easy. Given such a strategy $\sigma'$ in $M''$, we follow in 
$M'$ (and in $M$)
the same strategy (of course, 
only the first component of states of $M''$ matters in $M$),
until just before it transitions to a state $s_R$,
at which point it must be in $T_R$, and at that point our
strategy $\sigma$
switches to the strategy $\mu_R$.  This guarantees, for 
every $\varphi_i \in \Phi$, 
$Pr^\sigma_u(\varphi_i) \geq \Pr^{\sigma'}_{u^*}(\Diamond F_i) \geq r_i$.

For the other direction, suppose that the claim is not true,
i.e. there is a strategy $\sigma$ in $M$ which ensures
probability $\Pr^{\sigma}_u(\varphi_i) \geq r_i$, $i=1,\ldots, k$,
but $r \not \in U_{M''_{u^*},\overline{F}}$.
Note that all states in $F = \cup^k_{i=1} F_i$ are absorbing.
From
Theorem  \ref{thm:red-to-multi-LP}
we know that $U_{M''_{u^*},\overline{F}} = U_{M''_{u^*},\overline{F}}^{\odot}$  
where $U_{M''_{u^*},\overline{F}}^\odot$ is the
set of value vectors achievable by memoryless strategies.
Recall, 
that $U_{M''_{u^*},\overline{F}} = U_{M''_{u^*},\overline{F}}^\odot$ is convex, 
and that it is downward-closed.
Since $r \not\in U_{M''_{u^*},\overline{F}}$, 
as in the proof of (3.) $\Rightarrow$ (1.) in Thm. \ref{thm:red-to-multi-LP},
there must be a separating hyperplane, i.e., a non-negative weight vector 
$w= \langle w_1, \ldots , w_k \rangle$ such that
$w^Tr = \sum^k_{i=1} w_i r_i > w^Tx$ for every point $x \in U_{M''_{u^*},\overline{F}}$.

Consider $M$ with the following reward structure, denoted $rew(w)$:
a trajectory $\tau$ of $M$ receives reward 
$\sum \{ w_i \mid \tau ~{\rm satisfies}~\varphi_i \}$.
This is not the traditional type of reward structure where
reward is obtained at the states and transitions of
the trajectory; it is obtained only at infinity 
when the trajectory has finished and we get a reward
that depends on the properties that were satisfied.
In \cite{CY98} optimization of the expected reward
for MDPs with this kind of reward structure was studied
and solved by
reducing the problem to an MDP with a classical type of reward.
We reuse that construction here.
Consider the MDP $M''$ augmented with a traditional type of reward structure,
denoted $rew''$,
in which each transition of the form $(u,\gamma_R,1, s_R)$ 
produces reward $\sum \{ w_i \mid \varphi_i \in R \}$,
while all other transitions (and states and actions) give 0 reward.  
Let ${\hat M}''$ be a subMDP of $M''$ that
contains for each state $u$ only one (at most) transition
of the form $(u,\gamma_R,1,s_R)$, 
namely the one that produces the maximum reward (breaking ties arbitrarily).
Clearly, there is no reason ever to select from a state $u$ any
transition $(u,\gamma_{R'},1,s_{R'})$ that produces lower reward, thus,
$M''$ and ${\hat M}''$ have the same optimal expected reward.
It is shown in \cite{CY98} that the optimal expected rewards
in $(M,rew(w))$ and $({\hat M}'', rew'')$, and thus also in $(M'', rew'')$,
are equal to each other. Moreover, the optimum value
in these MDPs is achievable, i.e., there are optimal strategies,
and in fact  a
deterministic finite-memory optimal strategy can be constructed.

The optimal expected reward in $(M,rew(w))$ 
is at least $w^Tr$ (because strategy $\sigma$ achieves $w^Tr$),
while the optimal expected reward in $(M'', rew'')$ is 
equal to $\max \{ w^Tx \mid x \in U_{M''_{u^*},\overline{F}} \}$, 
because rewards are only
obtained by transitioning to a state in $F$.
Therefore, $w^Tr \leq \max \{ w^Tx \mid x \in U_{M''_{u^*},\overline{F}} \}$,
contradicting our hypothesis that $w^Tr > \max \{ w^Tx \mid x \in U_{M''_{u^*},\overline{F}} \}$. 
\end{proof}

It follows from the lemma that: there exists a strategy $\sigma$ in $M$
such that 
\[\bigwedge^{k'}_{i=1} \mbox{Pr}^\sigma_u(\varphi_i) \geq r_i
\wedge \bigwedge^{k}_{j=k'+1} \mbox{Pr}^\sigma_u(\varphi_j) > r_j \]
if and only if
there exists a strategy $\sigma'$ 
in $M''$ such that 
\[\bigwedge^{k'}_{i=1} \mbox{Pr}^\sigma_{u^*}(\Diamond F_i) \geq r_i
\wedge \bigwedge^{k}_{j=k'+1} \mbox{Pr}^\sigma_{u^*}(\Diamond F_j) > r_j\,.\]
Moreover, such strategies can be recovered from each other.
Thus (1.) and  (2.) follow, using Corollary \ref{cor:ach-and-pareto}.
\end{proof}

\vspace*{-0.23in}

\section{Concluding remarks}

\vspace{-0.08in}

We mention that recent results by Diakonikolas and Yannakakis
\cite{DY08} provide improved upper bounds for appoximation
of convex Pareto curves, and for computing a 
smallest such approximate convex Pareto set.
These results yield significantly improved algorithms,
particularly in the bi-objective case,
for the multi-objective LP problem,
and thus also for the multi-objective MDP problems studied in this paper.
In particular, in the bi-objective MDP case, \cite{DY08} provides 
a polynomial time algorithm
to compute a minimal $\epsilon$-approximate (convex) Pareto set
(i.e., one with the fewest number of points possible).

We mention that,
although  we use LP methods to
obtain our complexity upper bounds, 
in practice there is a way to combine other efficient iterative  
methods used for solving MDPs, e.g., based on value iteration
or policy (strategy) iteration,
with our results in order
to approximate the Pareto curve for multi-objective
model checking.  This is 
because the results of \cite{PapaYan00,DY08} for multi-objective convex
optimization problems only 
require a black-box routine that optimizes (exactly or
approximately) positive linear combinations of the objectives.
Specifically, in our setting the multiple MDP objectives 
ask to optimize the probabilities of different linear-time $\omega$-regular
properties.   By using the results in \cite{CY98},
it is possible to reduce the problem of optimizing such 
positive linear combinations to the problem of 
finding the optimal expected reward for a new MDP with positive rewards. 
The task of computing or approximating this optimal
expected reward can be carried out  
using any of various standard iterative 
methods, e.g., based on value iteration and policy iteration
(see \cite{Puterman94}).  These can thus be used to answer
(exactly or approximately) the black-box queries required by the methods
of \cite{PapaYan00,DY08}, thereby yielding a method for approximating
the Pareto curve (albeit, without the same theoretical complexity
guarantees).

An important extension of the applications of our results 
is to extend the asymmetric assume-guarantee compositional reasoning rule
discussed in Section 
\ref{sec:basics}
to a general compositional framework for probabilistic systems.  
It is indeed possible to describe symmetric assume-guarantee rules 
that allow for general composition
of MDPs.  A full treatment of the general compositional framework
requires a separate paper, and we plan to expand on this in follow-up work.

\vspace*{0.04in}

\noindent{\bf Acknowledgements.}  We thank the Newton Institute, 
where we initiated discussions on the topics
of this paper during the Spring 2006 programme on Logic
and Algorithms. Several authors acknowledge support from the following grants:
EPSRC GR/S11107 and  EP/D07956X, MRL 2005-04; 
NSF grants CCR-9988322, CCR-0124077, CCR-0311326,
and ANI-0216467, BSF grant 9800096, Texas ATP grant 003604-0058-2003,
Guggenheim Fellowship;  NSF CCF-04-30946 and NSF CCF-0728736.

\bibliographystyle{alpha}
\bibliography{/home/kousha/Bibliography/nice_bib}

\end{document}